\documentclass[a4paper,UKenglish,final]{llncs}
\usepackage{amssymb,amsmath,enumerate,showkeys}
\usepackage{fancyhdr}
\usepackage{stmaryrd}
\usepackage{xspace}
\usepackage{relsize}
\usepackage{color}
\usepackage{tikz}
\usepackage{ntheorem}

\renewtheorem{lemma}[theorem]{Lemma}
\renewtheorem*{lemma*}{Lemma}
\renewtheorem{proposition}[theorem]{Proposition}
\renewtheorem{corollary}[theorem]{Corollary}
\renewtheorem{example}[theorem]{Example}

\renewtheorem{proposition}[theorem]{Proposition}

\usepackage{prettyref}

\newrefformat{thm}{Theorem~\ref{#1}}
\newrefformat{lem}{Lemma~\ref{#1}}
\newrefformat{dfn}{Definition~\ref{#1}}
\newrefformat{cor}{Corollary~\ref{#1}} 
\newrefformat{prop}{Proposition~\ref{#1}}
\newrefformat{sec}{Section~\ref{#1}}
\newrefformat{rem}{Remark~\ref{#1}}
\newrefformat{fig}{Figure~\ref{#1}}
\newrefformat{eq}{Equation~\eqref{#1}}
\newrefformat{ex}{Example~\ref{#1}}
\newcommand{\prref}[1]{\prettyref{#1}\xspace}

\colorlet{DMnormalbackcolor}{gray!25}
\colorlet{DMlightbackcolor}{gray!10}
\colorlet{DMmediumbackcolor}{gray!20}
\colorlet{DMdarkbackcolor}{gray!30}

\definecolor{blue}{rgb}{0.211,0.211,0.656}

\definecolor{dgreen}{rgb}{0.0,0.4,0.0}

\newcommand{\red}{\color{red}}

\newcommand{\set}[2]{\left\{#1\mathrel{\left|\vphantom{#1}\vphantom{#2}\right.}#2\right\}}
\newcommand{\oneset}[1]{\left\{\mathinner{#1}\right\}}
\newcommand{\smallset}[1]{\left\{\mathinner{#1}\right\}}

\newcommand{\abs}[1]{\left|\mathinner{#1}\right|}

\newcommand{\N}{\mathbb{N}}
\newcommand{\Z}{\mathbb{Z}}

\newcommand{\Gama}{\ensuremath{\Gamma_\alpha}}
\newcommand{\Gamb}{\ensuremath{\Gamma_\beta}}
\newcommand{\Ga}{\ensuremath{G_\alpha}}
\newcommand{\Gb}{\ensuremath{G_\beta}}

\newcommand{\smalloverline}[1]
{{\mspace{1mu}\overline{\mspace{-1mu}#1\mspace{-1mu}}\mspace{1mu}}}
\newcommand{\ov}[1]{\smalloverline{#1}}

\newcommand{\oi}[1]{{#1}^{-1}}
\newcommand{\wh}[1]{\widehat{#1}}
\newcommand{\wt}[1]{\widetilde{#1}}

\newcommand\SL{\mathop\mathrm{SL}}

\newcommand{\IFF}{if and only if\xspace}
\newcommand{\homo}{homomorphism\xspace}
\newcommand{\homos}{homomorphisms\xspace}

\newcommand{\amal}{amalgamated\xspace}

\newcommand{\iso}{iso\-mor\-phism\xspace}

\newcommand{\logs}{logspace\xspace}

\newcommand{\WP}{word problem\xspace}
\newcommand{\CP}{conjugacy problem\xspace}

\newcommand{\mazu}{{M}azur\-kie\-wicz\xspace}

\renewcommand{\phi}{\varphi}
\newcommand{\eps}{\varepsilon}

\newcommand{\alp}{\alpha}

\newcommand\al{\mathop\mathrm{alph}}

\newcommand{\bet}{\beta}
\newcommand{\gam}{\gamma}
\newcommand{\del}{\delta}
\newcommand{\lam}{\lambda}

\newcommand{\Sig}{\Sigma}
\newcommand{\Gam}{\GG}

\newcommand\GG{\Gamma}

\newcommand{\Oh}{\mathcal{O}}

\newcommand{\vdmatrix}[4]{\begin{pmatrix}{#1}&{#2}\\{#3}&{#4}\end{pmatrix}}

\renewenvironment{proof}[1][Proof.]{
\begin{trivlist}
\item[\hskip \labelsep {\bfseries #1}]}{\hspace*{\fill}$\Box$\end{trivlist}
}

\newcommand{\bs}{\backslash}

\newcommand\lds{,\ldots ,} 
 
\newcommand{\sse}{\subseteq}

\newcommand{\sm}{\setminus}
\newcommand{\os}[1]{\oneset{#1}}

\newenvironment{vd}{\noindent\color{blue} VD }{}

\newenvironment{jk}{\noindent\color{red} JK }{}

\newenvironment{aw}{\noindent\color{magenta} AW }{}

\newenvironment{jl}{\noindent\color{blue} JL : }{}

\begin{document}
\title{Logspace computations in graph products} 
\author{Volker Diekert \and Jonathan Kausch}
\institute{
FMI, Universit\"at Stuttgart \\
Universit{\"a}tsstra{\ss}e 38\\
70569 Stuttgart, Germany \\ [1mm]
\email{\{diekert,kausch\}@fmi.uni-stuttgart.de} 
}

\maketitle

\begin{abstract}
We consider three important and well-studied algorithmic problems in group theory: the word, geodesic, and conjugacy problem.
We show transfer results from individual groups to graph products. 
We concentrate on \logs complexity because the challenge is actually in small complexity classes, only. The most difficult transfer result is for 
the \CP. We have a general result for graph products, but even in the special case of a 
graph group the result is new. Graph groups are closely linked to 
the theory of \mazu traces which form an algebraic model for concurrent processes. Our proofs are  combinatorial and based 
on well-known concepts in trace theory. 
We also use rewriting techniques over traces. For the group-theoretical part we apply Bass-Serre theory. But as we need explicit formulae and 
as we design concrete algorithms all our group-theoretical calculations are completely explicit and accessible to non-specialists. 
\end{abstract}



\section{Introduction}
{\noindent \bf Background.} Algorithmic questions concerning finitely generated groups have been studied for more than 100 years starting with the fundamental work 
of Tietze and Dehn in the beginning 
of the 20th century. In this paper we investigate three algorithmic problems for graph products $G$ with a finite and symmetric generating set $\Sig$. The question for us is whether they can be decided in \logs.
\begin{enumerate}
\item {\bf Word problem}. Let $w \in \Sig^*$. Is $w =1 $ in the group $G$?
\item {\bf Geodesic problem}. Let $w \in \Sig^*$. Compute a geodesic, i.e., a shortest word 
representing $w \in G$; and, if a linear order on $\Sig$ is defined, 
compute the lexicographical first word among all geodesics, i.e., compute 
a shortlex normal form of $w$.
\item {\bf Conjugacy problem}. Let $u,v \in \Sig^*$. Are $u$ and $v$ conjugated in $G$?
\end{enumerate}
The complexity of the first and third problem depends on $G$ only, whereas for the second problem we have to specify $\Sig\sse G$, too. Over the past few decades the search and design of algorithms for 
decision problems like the ones above has developed into 
an active research area, where algebraic methods and  computer science techniques  join in a
fruitful way, see e.g.~the recent surveys \cite{Lohrey2012survey,Sapir2011BullMath}. Of particular interest are those problems
which can be solved efficiently in parallel. 
More precisely, we are interested in {\em deterministic logspace}, called simply 
{\em  \logs} in the following. This is a complexity class at  the lower level in the 
$\mathsf{NC}$-hierarchy
\footnote{$\mathsf{NC}^i$ is the  class of languages which are  accepted by (uniform) boolean circuits
of polynomial size, depth $O(\log^i(n))$ and constant fan-in, see e.g.~\cite{Vollmer99} for a textboook.}:
\begin{equation}\label{eq:nc}
\mathsf{NC}^1 \subseteq \mathsf {\logs} 
 \subseteq {\mathsf {LOGCFL}}\subseteq
\mathsf{NC}^2 \subseteq \mathsf{NC}^3
\subseteq \cdots \subseteq \mathsf{NC} = \bigcup_{i \geq 1} \mathsf{NC}^i \subseteq \mathsf{P}\subseteq \mathsf{NP}
\end{equation}
No separation result between $
\mathsf{NC}^1$ and $\mathsf{NP}$ is known but it is believed (by some)
that all of the above inclusions in (\ref{eq:nc}) are strict. 
A fundamental result
in the context of group-theoretical algorithms
was  shown by  Lipton, Zalcstein and Simon
in \cite{lz77,Sim79}: The word problem of finitely generated linear groups 
belongs to logspace. 
The class of groups with a word problem in
logspace is further investigated in \cite{Waack81}.   
Another important result due to Cai (resp.~Lohrey) is that the word problem of hyperbolic groups is in  $\mathsf{NC}^2$ \cite{cai92stoc} (resp.~in {\sf LOGCFL} by \cite{Lo05ijfcs}). The class {\sf LOGCFL} coincides with the (uniform) class 
${\sf SAC}^1$. It is a subclass of $\mathsf{NC}^2$.
Often, it is not enough to solve the word problem, but one has to compute a normal form. This leads to the problem of computing geodesics. This problem and 
various related problems were studied e.g.~in 
\cite{DLS,elder2010,ElRe2010,MyRoUsVe08,PRaz}. These results imply that there are groups with 
an easy word problem (in logspace), but where simple questions related to 
geodesics are computationally hard, for example $\mathsf{NP}$-complete for certain wreath products 
or free metabelian group of rank 2.
Finally, the \CP is a classical decision problem which is notoriously more difficult than the \WP. 
Whereas for a wide range of groups the \WP is decidable (and often easily decidable) the conjugacy problem is not known to be decidable. 
This includes e.g.{} automatic groups (\WP is in $\Oh(n^2)$) or one-relator groups (\WP is decidable) to mention two classes. 
Miller's group \cite{Miller1} has a decidable word problem (at most cubic time\footnote{Mark Sapir, personal communication}, actually logspace) but undecidable \CP. Actually, there are finitely generated subgroups of $F_2 \times F_2$ (hence subgroups of  $\SL(4,\Z)$, hence linear groups) with unsolvable conjugacy problem \cite[Thm.{} 5.2]{Miller}.

Here, we continue and generalize the work of  \cite{dkl12conm} from graph groups to graph products of groups having a \WP in \logs.  
We show transfer results for all three problems mentioned above. However, techniques of \cite{dkl12conm} for graph groups
(which used linear representations for right-angled Coxeter groups) 
are not  available in the present paper, simply because linear representations do not exist for the individual groups, in general. 
For graph products we start with a list $L$ of groups $\Ga$. 
Next, we endow $L$ with an irreflexive and symmetric relation $I\sse L \times L$. This means $(L,I)$ is a finite undirected graph and each node $\alp \in L$ is associated with a node group $\Ga$.
The graph product $G$ is then the free product of the $\Ga$'s modulo defining relations $gh = hg$ for all $g\in \Ga$ and  $h\in \Gb$
where $(\alp, \bet) \in I$. Thus, it is a free product with partial commutation. 
If $I$ is empty then $G$ is the free product
$\star_{\alp \in L}\Ga$. If $(L,I)$ is a complete graph then 
$G$ is the direct product
$\prod_{\alp \in L}\Ga$.
Our setting includes the important special case where all node groups are isomorphic to $\Z$. This is exactly the case when $G$ is 
\emph{free partially commutative}. These groups are also known as 
\emph{graph groups} (see \cite{dro85}) or \emph{right-angled Artin groups} (RAAGs). If all node groups are isomorphic to $\Z/ 2 \Z$ then we  
obtain a \emph{right-angled Coxeter group}. 
Graph groups  embed into right-angled Coxeter groups in a canonical way \cite{hw99}, and Coxeter groups are known to be linear. 
Hence,
 graph  groups and Coxeter groups have a \WP in \logs. Graph groups received in recent years a lot of attention in group
theory because of their rich subgroup structure \cite{BesBr97,CrWi04
}. On the algorithmic 
side, (un)decidability results were obtained for many important
decision problems in graph groups 
\cite{CrGoWi09,dm06}. The theory of free partially commutative groups is also directly linked to the theory of \mazu traces which is important in computer science since it yields an algebraic framework for concurrent systems \cite{kel73,maz77,dr95}.  
\medskip 

{\noindent \bf Results.} Our achievement has a strikingly simple formulation: 
\noindent If the word problem (geodesic problem resp., \CP resp.) of all node groups is in \logs then the same is true for the graph product. 

\medskip
\noindent An analogous assertion holds for various other complexity classes  closed under \logs reductions like $\mathsf{NC}, \mathsf{P}$ or $\mathsf{NP}$ by similar arguments as used in this paper. We treat ``\logs'' because it  concerns the smallest natural complexity class where we can assert such a statement  because the 
\WP of non-abelian free groups  has to be solved, which  is $\mathsf{NC}^1$-hard by \cite{Robinson93phd}; and  
the best known upper bound is \logs. So, it is possible that 
\logs is truly the smallest class in all non-trivial cases.

Our results with respect to the \WP generalize in particular \cite[Prop.~19]{ElderEO13} solving thereby an open problem.  
A transfer result with respect to the \WP was known before for free products \cite{Waack81}, but unknown for graph products, in general. 
For a compressed variant of the word problem, it is known that polynomial time decidability is preserved by graph products \cite{HauboldLM12}.
Our results here imply that the \WP of a graph products of linear groups is solvable in \logs. 
It is still open whether a graph product of linear groups is linear again.
The results here support a positive answer to this question asked in \cite{hw99}; but not much beyond the classical result of \cite{wehrfritz73lms} is known.
Our method also yields a \logs-reduction  of the \CP for graph products of linear groups 
to the \CP in the node groups. 
This is somewhat the best we can expect because, as we mentioned above, there  
are finitely generated linear groups with unsolvable conjugacy problem.

Our proof is inductive on the number of nodes $\alp \in L$ and the algebraic description of a graph product as a certain amalgamated product. Amalgamated products are basic components in Bass-Serre theory\footnote{Bass-Serre theory is a cornerstone in modern combinatorial group theory. It  showed us the direction to the proof, but the abstract theory does not give  complexity results, directly.} 
\cite{serre80}; and indeed, our  proof of  \prref{cor:fitz}is an application of explicit Bass-Serre theory. The proof is still technical  and not easy. On the positive side we had to make all calculations explicit.
Thus, no a priori knowledge in Bass-Serre theory is necessary for understanding the \logs solution of the \WP in graph products. 

\section{Notation}
{\bf Words.} An \emph{alphabet} is a set (with a linear order) and  its elements 
are called \emph{letters}. By $\Sig^*$ 
we denote the free monoid over $\Sig$ and its elements are called \emph{words}. For a word $w\in \Sig^*$ we denote by $\abs w$ its \emph{length}
and if $a \in \Sig$, then ${\abs w}_a$ counts how often the letter $a$ appears in $w$. Thus, ${\abs w} = \sum_{a \in \Sig}{\abs w}_a$. 
By $\al(w) = \set{a \in \Sig}{{\abs w}_a \geq 1}$ we denote the 
\emph{alphabet} of $w$. The \emph{empty word} has length $0$; 
and it is denoted by 1 as other neutral elements in monoids or groups.
In the paper, $L$ is a finite  list (with a linear order) and 
$\Sig$, $\Sig_\alp$, $\Gam$,  $\GG_\alp$ denote alphabets.
We have $\Sig_\alp\sse \GG_\alp$, $\Sig = \bigcup_{\alp \in L}\Sig_\alp$
and $\Sig$ is finite, $\GG = \bigcup_{\alp \in L}\GG_\alp$ and 
$\GG$ is typically infinite. 
All alphabets are endowed with an \emph{involution}. This is a mapping 
$x\mapsto \ov x$ such that $\ov{\ov x} = x$. The involution is  extended 
to words by $\ov{a_1 \cdots a_n} = \ov{a_n} \cdots \ov{a_1}$ where 
$a_i$ are letters. For a group $G$ the involution is here always defined
by taking the inverse, i.e., $\ov g = \oi {g}$ for $g \in G$. 
As we represent group elements by words, we prefer the notation
$\ov g$ rather than $\oi {g}$ for group elements, too. 

{\noindent \bf Groups.} Our frame  is given by groups $\Ga$ (for $\alp$ in the list $ L$) which are assumed to be generated by some finite subset $\Sig_{\alp}$ with $\Sig_{\alp} = \Sig_{\alp}^{-1} \sse  \Ga\sm \os{1}$.
Moreover we define an alphabet $\GG_{\alp} =  \Ga\sm \os {1}$. Hence $\GG_{\alp}$ is infinite, in general. 
This means there is a natural surjective \homo from $\Sig_{\alp}^*$ onto $\Ga$ which respects the involution. Moreover, every letter  $a\in \GG_{\alp}$ can be represented by a word $w_a \in \Sig_{\alp}^*$
such that $w_a = a$ in $\Ga$. 

{\noindent \bf Graphs.} Here, graphs are without self-loops and multiple edges. 
They are node-labeled. The undirected graphs specify the ``independence''  relation. Directed graphs specify  ``dependence graphs'' which are used to represent group elements in graph products. 
We say that graphs are \emph{identical}, if they are isomorphic as node-labeled (directed) graphs. Thus, graphs are viewed as  ``abstract graphs''.

{\noindent \bf Complexity.} We use standard notation from complexity theory, \cite{papa,Vollmer99}. 
In particular, we use the result that the composition of \logs computable functions is \logs computable.
A function $f$ is computable in \logs if it is computable by some 
deterministic Turing machine such that the work tape is bounded by $\Oh(\log n)$ where $n$ denotes the input length. The output length 
is then bounded by some polynomial in $n$ and every \logs-computable
function is computable in $\mathsf{P}$, i.e. deterministic polynomial time. 

\section{Word problem in certain \amal products}\label{sec:amal}
Our results concern graph products. The results in this section serve 
as a tool during an induction process. They are slightly more general than needed there. 
The situation in this section is as follows. 
We consider finitely generated groups $A, B, P$ such that $A \leq P$ is a subgroup of $P$ and we let 
$G = P \star_A (A \times B)$ be the \amal product over $A$. 
The identity on $P$ and the projection of $A \times B$ onto  $A$ induce a projection $\pi: G \to P$. Let $H$ be the kernel 
 of $\pi$, then we have a short exact sequence
 $$1 \to H \to G \overset{\pi}{\to} P \to 1.$$
 Moreover, since $\pi(a) = a$ for $a \in A$, the \homo $\pi$ and the identity on $A \times B$
 induce a \homo of $G$ onto $P\times B$.  Thus, a necessary condition for an element to be $1$ in $G$ is that its image in $P\times B$ is $1$. 
 
\begin{theorem}\label{thm:wpinamal}
Assume that the \WP of $P$ and $B$ can be solved in \logs and that 
the membership problem for $A$ in $P$ can be solved in \logs, too. 
Then the \WP of $G$ can be solved in \logs.
\end{theorem}

\begin{proof}
 Let $w = g_0b_1 g_1 \cdots b_m g_m$ be a word 
 with $g_i \in P$ and $b_j \in B$. We want to decide whether 
 $w = 1 \in G$. In a first step, we simply compute 
 $\pi(w) = g_0 g_1 \cdots g_m \in P$ and we check whether 
 $\pi(w) =1$. This can be done in \logs and for the rest of the proof we may assume $\pi(w) =1$ (because otherwise $w \neq 1 \in G$) and hence 
 we have $w \in H$. Moreover, we may also assume that 
 $ b_1 \cdots b_m = 1\in B$.
 
 The structure of $H$ is well understood by Bass-Serre theory.
 The group $H$ is a free product of groups $ B^p = pB p^{-1}$ for certain $p \in P$. 
 For those readers who are familiar with Bass-Serre theory let us note that $H$  is the fundamental group of a  graph of groups which is a ``star'' with a  trivial group in the  center:

 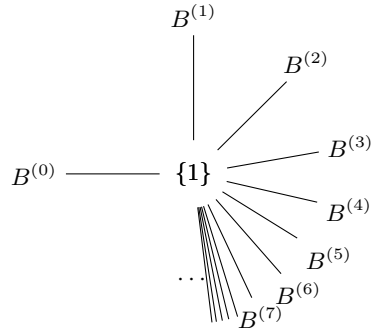
\begin{figure}[ht] 
 \begin{center}
\begin{tikzpicture}[scale=0.7]
\newcommand{\rad}{3}
\node[circle] (M) at (0,0 ) {$\smallset{1}$};
\node[circle] (L) at (0,0 ) {$\smallset{1}$};
\node (0) at (180:\rad) {$B^{(0)}$};
\node (1) at (90:\rad) {$B^{(1)}$};
\node (2) at (45:\rad) {$B^{(2)}$};
\node (3) at (10:\rad) {$B^{(3)}$};
\node (4) at (-13:\rad) {$B^{(4)}$};
\node (5) at (-32:\rad) {$B^{(5)}$};
\node (6) at (-49:\rad) {$B^{(6)}$};
 \node (7) at (-65:\rad) {$B^{(7)}$};
  \node (8) at (-73:\rad) {};
   \node (9) at (-76.5:\rad) {};
    \node (10) at (-79:\rad) {};
        \node (11) at (-81.3:\rad) {};
         \node (12) at (-83:\rad) {};
  \draw (0) --(M);
 \draw (1) --(M);
   \draw (2) --(M);
 \draw (3) --(M);
   \draw (4) --(M);
 \draw (5) --(M);
   \draw (6) --(M);
 \draw (7) --(M);
   \draw (8) --(M);
     \draw (9) --(M);
     \draw (10) --(M);
         \draw (11) --(M);
          \draw (12) --(M);
     \node at (270:\rad*2/3) {$\cdots$};
 \end{tikzpicture}
\end{center}
\caption{Star with $[P:A]$ rays, because $P \leq G$ induces a bijection
 $ P / A = H \bs G / (A\times B)$. }
\end{figure}

However, knowing the structure of $H$ is not enough, since we must be able to 
compute in \logs an effective  representation of $w$ in the free product. Moreover, $H$ is not finitely generated, in general. (This happens if $B$ is non-trivial and 
the index $[P:A]$ is infinite, which is the case of interest). So, instead of using Bass-Serre theory as a black box, we 
take a more elementary approach. 

Let $N$ be an index set and $\set{p_{\nu}\in P}{\nu \in N}$ 
be a subset of $P$ such that $\ov{p_{\mu}}p_{\nu} \notin A$ 
for all $\mu \neq \nu$. This means $\set{p_{\nu}A}{\nu \in N}$
is a set of pairwise disjoint cosets.  For each $\nu \in N$ we let 
$B^{(\nu)}$ be a copy of the group $B$. For $b^{(\nu)} \in B^{(\nu)}$
we let $b$ be the corresponding group element in $B$. 
Let $\psi_\nu: B^{(\nu)} \to H$ be the injective \homo defined by
$\psi_\nu(b^{(\nu)})= \ov {p_{\nu}}b p_{\nu}$. 
This induces a \homo $\psi: \star_{\nu \in N}B^{(\nu)} \to H$. Now, we have $H \leq G$ and since 
 $\ov{p_{\mu}}p_{\nu} \notin A$ 
for all $\mu \neq \nu$, a standard argument for \amal products 
shows that $\psi$ is injective.

Remember that we have $  w = g_0b_1 g_1 \cdots b_m g_m \in H$ 
with $g_i \in P$ and $b_j \in B$. Define and compute  $p_{i}= g_0 \cdots g_i \in P$ for $0 \leq i < m$. Then we have 
$w = p_0b_1 \ov{p_0} p_1b_2 \ov{p_1} \cdots p_{m-1}b_m \ov{p_{m-1}}$ because $\ov{p_{m-1}} = g_m$.
Thus, we can compute in \logs  for each $i$ the minimal index $\nu(i) \in \os{0 \lds m-1}$ 
   such that 
 $\ov{p_{\nu(i)}}\, p_i \in A$. Here we use that the membership problem for $A$ in $P$ is computable in \logs. 
Define  $N= \set{\nu(i)}{0 \leq i < m}$ and the \homo 
$\psi_\nu: B^{(\nu)} \to H$ as above. We obtain\footnote{What we have done so far is ``essentially''  a \logs reduction from the \WP in $G$ to the \WP in some free product $\star_{\nu \in N}B^{(\nu)}$. It is not a \logs reduction  in the literal sense, because $N$ depends on the input word, but on the positive side $\abs N \leq m$.} 
$w =  \psi(b_1^{(\nu_1) } \cdots b_m^{(\nu_m)})$ with $\nu_i = \nu(i)$.

Next, consider the homomorphism $\phi : \star_{\nu \in N} B^{(\nu)} \to  B$ where
$\phi(b^{(\nu)}) = b$.  We have  $\phi(\psi^{-1}(w))  = \phi(b_1^{(\nu_1) } \cdots b_m^{(\nu_m)}) = b_1 \cdots b_m\in B$. Recall that $ b_1 \cdots b_m= 1$, hence $\psi^{-1}(w)\in K$, where
$K = \text{ker}(\phi)$ denotes the kernel of $\phi$. 
Bass-Serre theory tells us that $K$ is free,  but we need to find and rewrite $\psi^{-1}(w)$ in some basis $X$ of some finitely generated free subgroup 
$F(X)\leq K$, so we do the explicit calculations.  

For simplicity of notation we may assume that the input word $w$ is written as $w = b{}_1^{(\nu_1)} \cdots b{}_m^{(\nu_m)}\in K$ with $m \geq 1$ and  $1 \neq  b{}_i\in B^{(\nu(i))}$. Since  $w \in K$, we have $m \geq 2$. We mimic what we have done above; and we define
$g_i^{(\ell)} = ({b{}_{1}}\cdots {b{}_{i}})^{(\ell)}\in B^{(\ell)}$.
In particular, $g_1^{(\ell)}= b_1^{(\ell)}$ and $g_m^{(\ell)}= 1$.
For each  $1\leq i <m$, consider the factor 
$b_{i}^{(k)}b_{i+1}^{(\ell)}$ of $w$ with $k = \nu_i$ and $\ell = \nu_{i+1}.$
 Replace $b_{i}^{(k)}b_{i+1}^{(\ell)}$ by  
$$b_{i}^{(k)}\; (\ov {b_{i}}^{(\ell)}\cdots  \ov {b{}_{1}}^{(\ell)})( {b{}_{1}}^{(\ell)}\cdots {b{}_{i}}^{(\ell)}) \, b{}_{i+1}^{(\ell)} = b{}_{i}^{(k)}\; \ov{g_{i}}^{(\ell)} g_{i+1}^{(\ell)}.$$
The input word $w$ becomes  (after this \logs procedure) a word of the form 
$w= g_1^{(\nu_1)} \ov{g_1}^{(\nu_2)}\, g_2^{(\nu_2)} \ov{g_2}^{(\nu_3)}\cdots g_{m-1}^{(\nu_{m-1})} \ov{g_{m-1}}^{(\nu_m)}$ with $g_i^{(\nu_i)} \ov{g_i}^{(\nu_{i+1})}\in K$.
We use the notation $ (i,g,j) = g^{(i)}\ov g^{(j)}\in K$ and we rewrite 
$w$ as a product over these triples with $1 \leq i,j <m $ and 
$g \in B$. The  triples define a subset of $K$ of size less than  $m^3$. 
 We have $(i,g,j)^{-1} = (j,g,i) \in K$.  But the set of $(i,g,j)$ is not a basis of $K$
since e.g., 
$(i,g,k)(k,g,j)= (i,g,j).$ In particular, we have $ (i,g,j) = (i,g,0)(0,g,j)$ for all  $ i ,j$. The next \logs computation rewrites 
$ w$ as a product in triples  $ (i,g,j) = g^{(i)}\ov g^{(j)}$
such that $ 1 \neq g \in B$ and $ i \neq j$, 
$ g^{(i)} \in B^{(i)}$ and  $ \ov g^{(j)} \in B^{(j)}$, 
 $ \phi(g^{(i)}) = g$  and $ \phi(\ov g^{(j)}) = g^{-1}$. 
Thus, we find in \logs a smallest set  $X = \set{(i,g,0)}{i \neq 0, g \neq 1}$ 
and a word $u \in (X \cup \ov X)^*$  such that $w = \xi(u)$
for the mapping $\xi: X \to K$, $\xi(i,g,0) = g^{(i)}\ov g^{(0)}$. 
The set $X$ has at most $m^2$  elements and 
the word $u$ can be viewed as an element in the free
 group $F(X)$.
 Standard \logs-computable encodings 
 embed $F(X)$ into the free group $F(x,y)$ with two generators. For example the 
 $i$-th generator in $x$ can be mapped to $x^i y x^{-i}$. 
 Another standard \logs-computable encoding embeds $F(x,y)$ into the special linear group $\SL(2,\Z)$ of 
 $2 \times 2$ integer matrices\footnote{For example, $x\mapsto \vdmatrix{0}{1}{-1}{-2}$ and $y \mapsto\vdmatrix{2}{-1}{1}{0}$.}. We can evaluate the matrix corresponding to $u$ in \logs by the Chinese remainder theorem. If the 
 evaluation is the identity matrix then we have $w=1$  (see \cite{lz77} for details).
 Thus, we may assume that the matrix is not the identity. Hence,
 $1 \neq u \in F(X)$. But this is not enough, we have to show that $1 \neq u \in F(X)$
 implies $w \neq 1$.
  The assertion  $\xi(u) = w \neq 1$  follows from the following lemma. 
  Thus, this lemma finishes the  proof of \prref{thm:wpinamal}.
  \end{proof}

\begin{lemma*}
The \homo $\xi: F(X) \to K$ is injective. This means $X$ forms a basis 
of a free subgroup of $K$ containing the element $w$.
\end{lemma*}

\begin{proof}We know $\xi(u) = w$. 
 Consider now any non-empty freely reduced word $u$ in $(X \cup \ov X)^*$
and let $\xi(u)$ be its image in  $K$.
We have to show that $\xi(u)\neq 1$. We can write  
$u= v\, (i,g,j)$, where  $v\in (X \cup \ov X)^*$ is a freely reduced word
and $(i,g,j) \in (X \cup \ov X)$. We show: 
\begin{itemize}
\item The last factor of $\xi(u)$ in the free product $\star_{\nu \in N} B^{(\nu)}$
is $\ov g^{(j)}$.
\item If $j=0$, then the last two factors of $\xi(u)$ are  $ h^{(i)}\ov g^{(0)}$ for some $ h\in B$.
\end{itemize}

For $ u = (i,g,j)$ we have $\xi(u)= g^{(i)}\ov g^{(j)}$ as desired. Hence, $v$ is not empty and we can write 
$u = v' (k,f,\ell) (i,g,j)$. By induction the last factor of 
$ \xi(v)$  is $\ov f^{(\ell)}$. 

For $\ell \neq i$ we conclude that the last three factors of $\xi(u)$ are  
$\ov f^{(\ell)} g^{(i)}\ov g^{(j)}$. Hence, we may assume that $\ell = i$. Therefore
$u = v' (k,f,i) (i,g,j)$.

For $f \neq g$ the last two factors of $\xi(u)$ are $(\ov fg)^{(i)}\ov g^{(j)}$.

Now, assume $f =g$, then we must have $k \neq j$ because $u$ is freely reduced. But then we must have $i= \ell =0$.
Therefore, 
$u = v' (k,g,0) (0,g,j)$ with $k \neq j$. 
By induction, the last two factors of $\xi(v)$ are $h^{(k)} \ov g^{(0)}$.
Hence, the last two factors of $\xi(u)$ are $h^{(k)} \ov g^{(j)}$.
In particular, $\xi(u)\neq 1$.
\end{proof}

\section{Graph products}\label{sec:grapro}
A graph product over groups is defined by the following data. There is a finite list $L$  and for each 
$\alp \in L$ there is an associated finitely generated non-trivial group $\Ga$.
In addition, there is an irreflexive symmetric relation $I \sse L \times L$, 
which is called  an \emph{independence relation}. This means, $(L,I; (\Ga)_{\alp\in L}))$ is a node-labeled undirected graph. The \emph{graph product } $G= G(L,I; (\Ga)_{\alp\in L}))$ is defined as the quotient group of the free 
product $ \star_{\alp \in L}G_\alp$ with defining relations 
${ g_\alp h_\bet =h_\bet g_\alp } \text{ for all }{ g_\alp\in G_\alp,\, h_\bet\in G_\bet },\,  { (\alp, \bet)\in I}.$ If all $\Ga$ are finitely presented, then the graph product $G$ is finitely presented, too. If the 
independence relation $I$ is empty, then $G$ is a free product. If
$(L,I)$ is a complete graph then $G$ is a direct product. The 
\emph{universal property} of $G$ is that a \homo of $G$ to another group 
$G'$ is given by a family of \homos $h_\alp:\Ga \to G'$ 
such that $h_\alp(x) h_\bet(y)= h_\bet(y)h_\alp(x)$ for all 
$(x,y) \in \Ga \times \Gb$ where $(\alp,\bet)\in I$. 

Graph products have an algebraic decomposition as in \prref{sec:amal}.
Start with any ``base'' node $\bet \in L$ and let 
$B = \Gb$. Consider the subgraph $(L',I')$ which is induced by $L \sm \os \bet$. This yields a corresponding graph product 
$P$. The \emph{link} of $\bet$ is the subgraph which is induced 
by the set of nodes $\alp \in L'$ where $(\alp, \bet) \in I$. 
Let $A$ be the graph product corresponding to the link of $\bet$. 
Then $A$ is a subgroup of $P$ and $A \times B$ is a subgroup of $G$.

\begin{example}\label{ex:gp}
Consider a graph product $G$ as depicted as follows. 
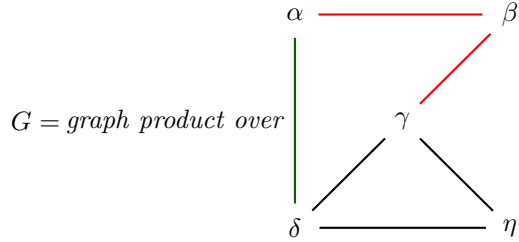
\begin{center} 
\begin{tikzpicture}[node distance=20mm]
 \node[circle] (a4) {$ \gam$};
 \node[circle] (a3) [below left of=a4] {$ \del$};
 \node[circle] (a2) [above left of=a4] {$ \alp$};

 \node[circle] (a5) [below right of=a4] {$ \eta$};
 \node[circle] (a6) [above right of=a4] {$ \bet$};
 
 \node[circle] (a1) [below left  of=a2] {  {\!\!\!\!\!\!\!\! graph product over} };
\node[circle] (a0) [left  of=a1] {$ G =  $};

 \draw (a2)  edge[-,thick,dgreen] (a3); 
 \draw (a2)  edge[-,thick, red] (a6);
 \draw (a3)  edge[-,thick] (a4); 
 \draw (a3)  edge[-,thick] (a5);
 \draw (a4)  edge[-,thick] (a5); 
 \draw (a4)  edge[-,thick, red] (a6); 
\end{tikzpicture}
\end{center}
The link of $\bet$ is $\os{\alp, \gam}$ and $A$ is the free product $ A = G_\alp \star G_\gam$. Removing the node $\bet$ 
we obtain a smaller graph and the link of $\alp$ becomes the singleton 
$\os {\del}$. Removing $\alp$ leaves us with a triangle with 
nodes $\gam$,  $\del$, $\eta$ which yields the direct product 
$G_\gam \times G_\del\times G_\eta$. Going backwards we see that 
$P$ is the amalgamated product $P= ({ G_\alp \times G_\del})
\star_{G_\del}{ (G_\gam \times G_\del\times G_\eta)}$ which contains $A$.
Finally, $G = { (G_\bet \times A )} 
\star_{A}P$. 
\end{example}

%


\begin{proposition}\label{prop:decomp}
The natural inclusions of $P$ and of $A \times B$ into $G$ induce 
an \iso between $P \star_A (A \times B)$ and $G$. 
\end{proposition}

\begin{proof}
 Trivial. Both sides satisfy the same universal property.  
\end{proof}


\begin{corollary}\label{cor:fitz}
The \WP of a graph product $G= G(L,I; (\Ga)_{\alp\in L}))$ 
is solvable in \logs \IFF the \WP of all node groups $\Ga$ is in 
\logs. 
\end{corollary}

\begin{proof}
 If the \WP of  $G$ is in \logs then the same holds for all finitely generated subgroups. For the other direction we write  
 $G$ as $P \star_A (A \times B)$ according to \prref{prop:decomp}.
 Now,  
 if $\pi_A:P \to A$ denotes the natural 
 projection which is the identity on $A$ and sends all elements outside $A$ to $1$, then we have $\pi_A(w) =w \iff w \in A$. Thus, the 
 membership problem ``$w\in A$?'' reduces in \logs to 
 an instance of the  \WP ``$w = \pi_A(w)$?'' in $P$.
 By induction, the \WP of  $P$ is solvable in \logs. Hence, \prref{thm:wpinamal} yields the result.  
\end{proof}

\subsection{Dependence graph representation}\label{sec:}
In order to represent elements in a graph product we use its 
dependence graph representation which was first introduced by \mazu 
in trace theory for free partially commutative monoids \cite{maz77}.
This representation takes the complement relation 
$D= L\times L \sm I$ into account. The relation $D$ is called 
\emph{dependence relation}. The idea is that it is enough to ``remember'' the ordering between dependent letters; an idea  which actually goes back to Keller
\cite{kel73}.


%

%


%

%
%

We use $ \Gama= \Ga \sm \os{1}$ as (infinite) alphabets for representing group elements in the groups $\Ga$. A word in $ \Gama^*$ denotes 
in a natural way an element in  $\Ga$. The empty word denotes 
$1 \in \Ga$, all other elements of $\Ga$ have a representation as a single letter. If $a_1\cdots  a_n \in \Gama^*$ is a word then we denote by 
$[a_1\cdots  a_n]$ the corresponding element in $\Gama \cup \os 1$ 
such that $a_1\cdots  a_n= [a_1\cdots  a_n]$ in the group $\Ga$. 
In addition, we let 
 $ \Gam$ be the disjoint union over all $ \Gama$ where $  \alp\in L$.
 In concrete algorithms we cannot work with $ \Gam$ directly. 
 Instead we use for each $\Gama$ a finite subset $\Sigma_\alp = \Sigma_\alp^{-1} \sse \Gama$ such that $\Sigma_\alp$ generates 
 $\Ga$. We let $\Sig\sse \Gam $ be the the  union over all $\Sigma_\alp$.
 The way we represent words $w$ over $\Gam^*$ is to write 
 them with brackets $w = [u_1]\cdots [u_n]$ where 
 each $u_i$ is a word in $\Sigma_\alp^*$ for some $\alp \in L$. If $u_i \neq 1 \in \Ga$ then $[u_i]$ becomes a letter of $\Gam$. Since we work only with graph products where the \WP in node groups 
 is solvable in \logs, we may always assume that $u_i \neq 1 \in \Ga$.
 Thus, $w = [u_1]\cdots [u_n]$ is a word in $\Gam^*$ of length $\abs w = n$. We start with an input word in $\Sigma^*$ and  an initial bracketing is somewhat arbitrary as long as we group only letters from one 
 $\Sigma_\alp$ together.

 Assume  $A\cup B \sse L$ such that $A\times B \sse I$. Let
 $\Gam_A = \bigcup_{\alp \in A} \Gama$ and $\Gam_B = \bigcup_{\bet \in B} \Gamb$. Then we call words $u \in  \Gam_A^*$ and $v \in \Gam_B^*$ 
 \emph{independent}. They can be shuffled into each other without changing the image in $G$. In particular, if $u$ and $v$ are independent then 
$uv=vu$ in $G$. Thus, independence implies commutativity in $G$, but the converse is false because the independence relation $I$ is irreflexive. 
This is a subtle but important feature to have unique normal forms in the 
graph representation. As a special case, let 
 $\bet \in L$ and denote  $I(\bet)=\left(\bigcup_{(\alp,\bet)\in I} \Gama\right)^*$. Then   $u \in \Gamb^*$ and $v\in I(\bet)$ are examples of independent words.
 For a word $ w = a_1 \cdots a_n\in \Gam^*$ we define a node-labeled acyclic graph $D(w) = [V,E,\lam]$, its \emph{dependence graph}, as follows:
\begin{itemize}
	\item 
	The vertex set $V$ is $ \os{1 \lds n}$.
	\item The label $\lam(i)$ of a vertex $i$ is the letter $ a_i \in G_{\alp_i}$. 
 	\item Arcs are from $ i$ to $ j$ for $ i <j$ where
	labels $\lam(i)$, $\lam(j)$ are dependent.	Thus, $E= \set{(i,j)\in V \times V}{i <j \wedge (\alp_i,\alp_j) \in D}$.						\end{itemize}
We view $D(w)$ as an abstract graph. This means we let $D(w) = D(w')$, 
if $D(w)$ and $D(w')$ are isomorphic as node-labeled directed graphs. 
For words $w, w' \in \Gam^*$ we write $w \equiv w'$, if $D(w)$ and $D(w')$ are isomorphic. For example, if $a\in \Gama$ and $b \in \Gamb$
with $(\alp, \bet) \in I$, then $ab \equiv ba$ and $D(ab) = D(ba)$. 
If $a,a'\in \Gama$ then   $D(aa')$ has two vertices and one arc, but
$D([aa'])$ has at most one vertex, hence $D(aa') \neq  D([aa'])$
and $aa' \not\equiv [aa']$.

Given an abstract graph $D(w) = [V,E,\lam]$
we associate to it a group element $g \in G$ as follows. 
We choose a topological sorting of $V$, this means we 
identify $V$ with $ \os{1 \lds n}$ such that $(i,j)\in E$ implies 
$i < j$. Then we let $g = \lam(1) \cdots \lam(n) \in G$. It is easy to see 
by induction on $n$ that $g$ is well-defined. 
The graph $D(w)$ can be reconstructed by its \emph{Hasse diagram}. 
The Hasse diagram removes all transitive edges. This means, we remove arc $(i,k)$ from $E$ as soon as there are 
$(i,j),(j,k)\in E$. The advantage of the Hasse diagram is that it is much smaller. For example the outdegree of every node is bounded by $\abs L$ whereas the outdegree of a node in $D(w)$ can be $\abs w -1$. 

The following rewriting procedure on dependence graphs relies on their  Hasse diagrams.
Let $D(w)$ be a dependence graph of some word $w \in \Gam^*$. 

{\noindent \bf Rewriting procedure:}
As long as  for some $\bet \in L$ there is an arc in the Hasse diagram  from $ i$ to $ j$ 
with labels $b$, $b'\in \Gamb$ do the following:
\begin{itemize}
\item Multiply $b\cdot b' = [bb']$ in 
$ \Gb$. 

	\item If $[bb'] = 1$ then remove vertices $ i$ and $ j$ and their incident arcs. 
\item If $[bb'] \neq  1$ then remove vertex $j$ and its 
incident arcs.
 Relabel vertex $ i$ by the letter $[bb']\in \Gamb$. 
	\end{itemize}
	The procedure transforms a dependence graph into a dependence graph
	with less vertices, but it
	does not change the corresponding group element in $G$. The rewriting procedure terminates in at most $\abs w$ steps. It yields a 
dependence graph $D(\wh w)$ with the property that 
labels of neighbors in the Hasse diagram belong to different 
nodes in $L$. A dependence graph with this property is called \emph{reduced}. A word $w \in \Gam^*$ is called \emph{reduced} if its
dependence graph  is reduced. 
We use the following characterization in order to check that 
a word $w$ and its dependence graph $D(w)$ are reduced.

\begin{lemma}\label{lem:charred}
A word $w \in \Gam^*$ is reduced \IFF it does not 
contain any factor $bub'$ such that $b,b' \in \Gamb$ and 
$u \in I(\bet)$ where $\bet \in L$.  
\end{lemma}

\begin{proof}
If a factor $bub'$ with $b,b' \in \Gamb$ and 
$u \in I(\bet)$ appears in $w$ then vertices $i$ and $j$ corresponding to 
the letters $b$ and $b'$ are neighbors in the Hasse diagram of $D(w)$ which is therefore not reduced. For example, in \prref{fig:fred} the dependence graph of the word $ ab\ov aca \ov b$ on  the left is not reduced but  the dependence graph of the word $ bc \ov ba$ on  the right is reduced.
\begin{figure}[h]
\begin{center} 
\begin{multicols}{2}{
\begin{tikzpicture}[node distance=20mm]
 \node[circle] (a4) {$ c$};
 \node[circle] (a3) [below left of=a4] {$ \ov a$};
 \node[circle] (a2) [above left of=a4] {$ b$};
 \node[circle] (a1) [left of=a3] {$  a$};

 \node[circle] (a5) [below right of=a4] {$ a$};
 \node[circle] (a6) [above right of=a4] {$ \ov b$};
 \node[circle] (= ) [below right of=a6] {is equal in $G$ to };

 \draw (a1)  edge[->,thick, red] (a3); 
 \draw (a1)  edge[->,bend left] (a4); 
 \draw (a1)  edge[->,bend right] (a5);
 \draw (a2)  edge[->,thick, red] (a4); 
 \draw (a2)  edge[->] (a6);
 \draw (a3)  edge[->,thick, red] (a4); 
 \draw (a3)  edge[->] (a5);
 \draw (a4)  edge[->,thick, red] (a5); 
 \draw (a4)  edge[->,thick, red] (a6); 
\end{tikzpicture}}

{
\begin{tikzpicture}[node distance=20mm]
 \node[circle] (a4) {$ c$};
 \node[circle] (a2) [above left of=a4] {$ b$};

 \node[circle] (a5) [below right of=a4] {$ a$};
 \node[circle] (a6) [above right of=a4] {$ \ov b$};

 \draw (a2)  edge[->,thick, red] (a4); 
 \draw (a2)  edge[->] (a6);
 \draw (a4)  edge[->,thick, red] (a5); 
 \draw (a4)  edge[->,thick, red] (a6); 
\end{tikzpicture}}
\end{multicols}
\end{center}
\caption{Dependence graphs ({Hasse diagrams in {\red red})} of $ ab\ov aca \ov b$ and $ bc \ov ba$.}\label{fig:fred}
\end{figure}
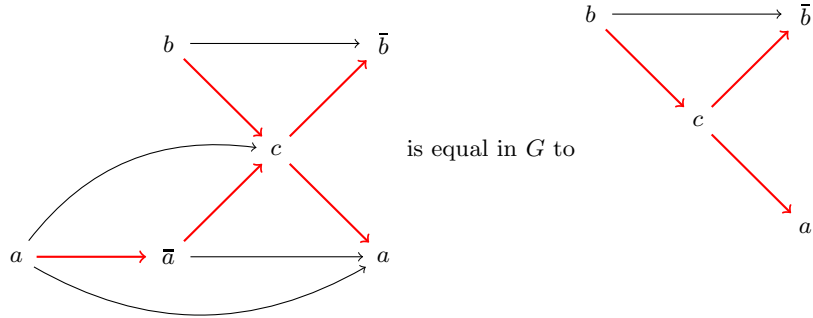

For the other direction, assume that 
no such factor appears. Let $w = a_1 \cdots a_n$ and $[V,E,\lam]$ be
its dependence graph. Consider $1 \leq i < j \leq n$ 
such that $\lam(i), \lam(j) \in \Gamb$ for some $\bet \in L$. 
Then there is some $i < k < j$ such that $\lam(k) \in \Gama$ and 
$(\alp, \bet) \in D$. Hence, $D(w)$ is  reduced.
\end{proof}

If a word $w$ is reduced, then its dependence graph is a 
unique normal form for the corresponding element in the graph product
$G$. This follows from the following technical result. 

\begin{proposition}[\cite{dl08ijac}]\label{prop:dieloh}
The rewriting procedure is confluent and yields normal forms for group elements in the graph product. In particular, reduced dependence 
graphs are isomorphic \IFF the associated group elements are the same. 
\end{proposition}

In the following, if  $w \in \Gam^*$ then $[w]$ denotes a 
reduced word such that $w=[w] \in G$. The dependence graph $D([w])$
is uniquely defined by $w$ (up to isomorphism). The normal form is therefore 
$D([w])$ rather than the reduced word $[w]$. Note that the notation 
$[w]$ is a generalization of the notation $[a_1 \cdots a_n]$ used above. 
\prref{prop:dieloh} reduces  the \WP of $G$ to the \WP{}s 
of the node groups as follows: 
We start with a word $w \in \Gam^*$. In order to run the rewriting procedure on $D(w)$, we just have to decide \WP{}s in node groups $\Gb$. 
At the end we have $w=1 \in G$ \IFF the procedure terminates in the empty graph. It is however far from obvious that we can implement the procedure 
in \logs. The following theorem is crucial.
It uses \prref{cor:fitz} as a black box. 

\begin{theorem}\label{thm:compred}
Let $G = G(L,I; (\Ga)_{\alp\in L})$ be a finitely generated graph product. Assume that the \WP for each node group $\Ga$ is in \logs. Then 
there is a \logs computation which transforms an input word $w$ over generators into a reduced dependence graph for $w$. 
\end{theorem}

\begin{proof}
The proof ends after  
\prref{lem:proc}. The result does not depend on the choice of generators. Therefore, we may assume $w \in \Sig^*$ where $\Sig = 
\bigcup_{\alp \in L} \Sig_\alp$ and each $\Sig_\alp = \Sig_\alp^{-1} $ generates $\Ga$. Every letter of the countable alphabet $\Gam$ is represented internally by some word in $\Sig^*$. Moreover, as 
$\Sig \sse \GG$ the input $w$ is a product of letters over $\Gam$. 
We perform $ \abs L$ rounds of \logs reductions. 
In each round  we minimize the number of letters $ a_i \in \Gama$
for some $\alp \in L$. For this we introduce the following notion. 
We say that $ w= u_0a_1u_1 \cdots a_nu_n$ is the 
\emph{$\alp$-factorization} of a word $w\in \Gam^*$, if we have 
$ a_i \in \Gama$ for $1 \leq i \leq n$ and $u_i \in (\Gam\sm \Gama)^*$
for $0\leq i \leq n$. The $\alp$-factorization exists and it is unique. 
We use the following procedure. 

{\noindent \bf The $\alp$-reduction:} Let
$ w= u_0a_1u_1 \cdots a_nu_n$ be  its $\alp$-factorization. For $n = 0$ do nothing. 
For $n > 0$ start with $i = 1$. 
\begin{itemize}
\item
From left-to-right: Stop at $ a_i$. Compute  the maximal $ m\geq i$ 
(by calling instances of the \WP in $G$) such that 
$$ a_i u_i \cdots  a_{m} u_{m}= 
a_{i}  \cdots a_{m}u_i \cdots u_{m} \in G$$
\item 
Replace $ a_i u_i \cdots  a_{m} u_{m}$ by $a u_i \cdots u_{m} $
where $a = [a_{i}  \cdots a_{m}]\in \Ga = \Gama \cup \os 1$. 
\item 
If $ m= n$ then the $\alp$-reduction is finished, otherwise 
change $i$ to $m+1$, stop there and continue the left-to-right phase.  
\end{itemize}

The overall  procedure performs $\alp$-reductions for all $\alp \in L$ in any order. The output can be read as a word $\wh w$ over $\Gam$. Reading from left to right we compute in a final round the actual dependence graph $D(\wh w)$. Each $\alp$-reduction can be done in \logs
due to \prref{cor:fitz}. Since there are only a constant number of rounds 
the overall  procedure is in \logs, too. It remains to prove that 
the output $D(\wh w)$ is the reduced dependence graph corresponding to the
input word $w$. The proof will be based on the next two lemmata. 

We define a word $w \in \Gam^*$ to be $\alp$-reduced, if 
every other word $w'\in \Gam^*$ having less letters from $\Gama$ denotes 
a different group element in $G$. Due to \prref{prop:dieloh} a word $w$ 
is reduced \IFF it is $\alp$-reduced for all $\alp \in L$.

\begin{lemma}\label{lem:geo}
Let $w = u_0a_1u_1 \cdots a_nu_n\in \GG^*$ be  its $\alp$-factorization. Then $w$ is $\alp$-reduced
\IFF $a_iu_ia_{i+1} \neq a_i a_{i+1}u_i \in G$ for all $1 \leq i < n$. 
\end{lemma}

\begin{proof}
If $a_iu_ia_{i+1} = a_i a_{i+1}u_i \in G$ for some $1 \leq i < n$, then 
$w$ is not $\alp$-reduced. Hence, it is enough to show that if  $a_iu_ia_{i+1} \neq a_i a_{i+1}u_i \in G$ for all $1 \leq i < n$ then  $w$ is $\alp$-reduced. This is true, 
if $w$ is reduced. Hence we may assume that $w$ is not reduced. Then there exists $\bet \in L$ and a  
factor $bub'$ with $b,b' \in \Gb$ and $u \in I(\beta)$. Since $a_iu_ia_{i+1} \neq a_i a_{i+1}u_i$ we must have $\alp \neq \bet$. If the factor $bub'$ is a factor inside 
some $u_i$, then we can rewrite it by $[bb']u$ and we obtain a word $w'$ which satisfies the same property, but which  length over $\GG$ is shorter. Hence by induction on the length
$w'$ is $\alp$-reduced. This implies that $w$ is $\alp$-reduced, too.

Thus we may assume that for some $i < j$ we have 
$u_i= p_ibq_i$ and $u_j= p_jb'q_j$ 
with $q_ia_{i+1}$, $a_jp_j \in I(\bet)$. 
In particular, $u_ia_{i+1} = p_iq_ia_{i+1}b \in G$ and we have $w = w'$ in $G$ where 
$w' = u_0a_1u_1 \cdots u_{i-1}a_{i} p_iq_ia_{i+1}bu_{i+1} \cdots a_{n}u_{n}$. 
 By induction on $\abs {j-i}$ we obtain that $w'$ is $\alp$-reduced. This implies again that $w$ is $\alp$-reduced.
\end{proof}

Let $w = u_0a_1u_1 \cdots a_nu_n\in \GG^*$ be  its $\alp$-factorization
and $0 \leq i \leq n $. We say that $u_0a_1u_1 \cdots a_iu_i$ 
is an $\alp$-prefix, if there are no $0<k<\ell$ such that $k \le i$ and $a_{k} u_{k} \cdots  a_{\ell}u_{\ell} = 
[a_{k}  \cdots a_{\ell}] u_{k} \cdots  u_{\ell} \in G$ with 
$\ell \leq n$. Note that $u_0$ is an $\alp$-prefix.
Moreover, $w$ is an $\alp$-prefix of itself \IFF $w$ is $\alp$-reduced
by \prref{lem:geo}.

\begin{lemma}\label{lem:proc}
Let $w = u_0a_1u_1 \cdots a_nu_n\in \GG^*$ be  its $\alp$-factorization
and $0 \leq i < n $ such that $u_0a_1u_1 \cdots a_iu_i$ 
is an $\alp$-prefix and let $m$ be maximal such that \\
$a_{i+1} u_{i+1} \cdots  a_{m }u_{m} = 
[a_{i+1}  \cdots a_{m }] u_{i+1} \cdots  u_{m} \in G$. 

Then $u_0a_1u_1 \cdots a_iu_i [a_{i+1}  \cdots a_{m }] u_{i+1} \cdots  u_{m}$ 
is an $\alp$-prefix of \\ $u_0a_1u_1 \cdots a_iu_i [a_{i+1}  \cdots a_{m }] u_{i+1} \cdots  u_{m}a_{m+1} u_{m+1} \cdots  a_{n} u_{n}$. 
\end{lemma}

\begin{proof}
 Straightforward since $i<m\leq n$  and $m$ was chosen to be maximal.
 \end{proof}

{In order to finish the proof of \prref{thm:compred} it  is enough to show that the \logs procedure ``$\alp$-reduction'' computes an 
 $\alp$-reduced word. The invariant of the procedure is that in the left-to-right phase  $\alp$-prefixes
are computed. This follows from \prref{lem:proc}. At the end of the $\alp$-reduction the
word $w$ itself becomes an $\alp$-prefix. But then \prref{lem:geo} tells us that $w$ is an $\alp$-reduced word.  Thus, if after an  $\alp$-reduction we perform a $\bet$-reduction then the word is
$\alp$- and $\bet$-reduced, and so on. Hence, the result of \prref{thm:compred}.} 
\end{proof}

\prref{thm:compred} implies the following result.

\begin{corollary}\label{cor:linconjug}
The \WP of a graph product of linear groups 
is solvable in \logs.
\end{corollary}

 The next corollary shows that shortlex normal forms can be computed 
 in \logs if this is possible for all node groups. In the statement of the result we assume that we have a linear order on the set of nodes $L$ and that each node group $\Ga$ is finitely generated by some 
 linearly ordered set $\Sig_\alp$. (For simplicity assume $\Sig_\alp= \Sig_\alp^{-1}$.)  We use 
 $\Sig = \bigcup_{\alp \in L}\Sig_\alp$ as a generating alphabet for 
 $G= G(L,I; (\Ga)_{\alp\in L}))$. The linear order is as follows. 
 If $\alp < \bet$ then every letter in $\Sig_\alp$ is before 
 $\Sig_\bet$. If $\alp = \bet$ then we use the order in $\Sig_\alp$. 
 The \emph{shortlex order} on $\Sig^*$ is defined as usual: 
 If $\abs u < \abs v$ then $u<v$ in the shortlex order. If 
 $\abs u = \abs v$ and $u$ is lexicographically before $v$ then $u<v$ in the shortlex order.
 The \emph{shortlex normal form} of an element $g$ is the unique 
 minimal word $w \in \Sig^*$ which satisfies $w = g$ in $G$. 
 The following corollary is an immediate consequence of \prref{thm:compred}. It generalizes the 
 main result in \cite{dkl12conm}.
\begin{corollary}\label{cor:compred}
Let $\Sig$ and  the graph product $G= G(L,I; (\Ga)_{\alp\in L}))$ as above. If for each node group $\Ga$ the shortlex normal form is computable in \logs, then the shortlex normal form in $G$ is computable in \logs.
\end{corollary}
 
 The following proposition will be used in the next section for solving conjugacy in graph products. Due to \prref{thm:compred} it  generalizes \prref{cor:fitz}. It states that 
 ``pattern matching'' over dependence graphs is possible in \logs. 
 
\begin{proposition}\label{prop:patmat}
Let $G= G(L,I; (\Ga)_{\alp\in L}))$ be a graph product 
such that  the \WP of all node groups $\Ga$ is in 
\logs. 
Then the following problem can be solved in \logs. 
Input: Words $p$, $t \in \Gam^*$. Problem: 
Do  $x,y$ exist such that $t \equiv  xpy$?
\end{proposition}

\begin{proof}We may assume that $1 \leq \abs p$ and ${\abs p}_\alp \leq {\abs t}_\alp$ for all $\alp \in L$.  First, we compute $D(t) = [V,E,\lam]$
and $D(p) = [V',E',\lam']$. 
 Let $\abs t= \abs V = n$ and $M'\sse V'$ be the set of minimal vertices in $D(p)$. There are at most 
 $n^{\abs{ M'}}$ positions in $t$ which may correspond to $M'$. 
 The \logs procedure may investigate each of them one after another. 
 So we may think that a copy $M\sse V$ of $M'$ is fixed. 
 In a next round we keep in  $V$ only those vertices $u$ which can be reached 
 by a directed path from some vertex in $M$. All other vertices are deleted. 
 Now we run a symmetric procedure with maximal vertices.  After that we may assume that the sets of minimal vertices of $p$ and $t$ and the sets of maximal vertices of
   $p$ and $t$ coincide. However, now there are $x,y$ such that $t \equiv  xpy$ \IFF $D(t) = D(p)$. This can be checked in \logs because the \WP of all node groups $\Ga$ is in 
\logs. Details are left to the reader.
\end{proof}

\section{Conjugacy} \label{sec:conjugacy}
Two group elements $u,v \in G$ are \emph{conjugate} if there exists a $z \in G$ with $z^{-1}uz = v$ in $G$. If $u$ and $v$ are conjugate we write 
$u \sim v$.
The conjugacy problem is to decide on input 
words $u,v$ whether or not $u \sim v$ as elements of $G$. The aim of this section 
is to prove the following result.

\begin{theorem}\label{thm:conjug}
The \CP of a graph product 
is solvable in \logs \IFF the \CP of all node groups 
is in 
\logs. 
\end{theorem}

The easy direction of \prref{thm:conjug} is the implication from left to right, because
for all $u\in \Ga$ 
and all $z\in G$ where the reduced dependence graph contains a vertex with a label in some $\Gb$ with 
$(\alp, \bet) \in D$ and
$\alp \neq \bet$ 
we have $zu\oi{z} \notin \Ga$. 
Hence, for $u,v\in \Ga$ we have $u\sim v$ in $\Ga$ if and only if $u\sim v$ in $G$.
Thus, we have to show only that if the \CP of all node groups $\Ga$ 
is in 
\logs then the \CP of  $G$ is in \logs.
Before we prove the other direction  let us state 
an immediate consequence of \prref{thm:conjug}. 
It solves an  open problem for  two
prominent classes of finitely generated groups.

\begin{corollary}\label{cor:conjug}
The \CP of a graph group or a right-angled Coxeter group can be solved in \logs.\end{corollary}

The proof of \prref{thm:conjug} (and its corollaries) covers the rest of this section. The \logs algorithm can be found at the end of this section, too. 
 Using \prref{thm:compred} we may assume that the input words
 $u,v\in \Gam^*$ are 
 reduced. Actually we work with their dependence graph representations. 
 Therefore it is convenient to have a special notation. 
 We write $w \equiv w'$ if $D(w)$ and $D(w')$ are isomorphic. 
 Recall that, if $w\in \Gam^*$ is reduced, then  $w \equiv w'$ if and only if 
 $w=w'$ in the graph product $G$. If $ w= u_0a_1u_1 \cdots a_nu_n$ is the 
$\alp$-factorization then we call $n$ the $\alp$-length. We denote it 
by $\abs{w}_\alp$. Thus $\abs{w}_\alp$ is the number of vertices in the 
dependence graph of $w$ having a label in $\Gama$.  For later use we also define the \emph{alphabet } of a word $w\in \Gam^*$  by
$\al(w) = \set{\alp \in L}{\abs{w}_\alp \geq 1}$. If $w$ is reduced, then 
it depends on the image in $G$, only. Thus, we can define the 
alphabet of group elements, too. We also say that 
a word is \emph{connected} if $\al(w)$ induces a connected subgraph in the dependence graph $(L, D)$. Assume that 
$\al(w)$ splits into connected components $A_1 \cup \cdots \cup A_k$. 
Then we have $w \equiv w_1 \cdots w_k$ with 
$\alp(w_i) = A_i$ and $w_iw_j \equiv  w_jw_i$ for all $1 \leq i < j \leq k$. If $w$ is reduced, then all $w_i$ are reduced. Therefore we can split every group element of $G$ into connected components which commute pairwise. 
Next, we use the following fact. 
 
\begin{lemma}\label{lem:cyc}
Let $u\in \Gam^*$ be reduced. Then there exists a unique minimal $\wt u$  such that $u \equiv p \wt u \ov p$ for some $p \in \Gam^*$.
\end{lemma}

\begin{proof}
 If there is no $a \in \Gama$ such that $u \equiv a u_1 \ov a$ then 
 we must choose $p=1$. Otherwise we rewrite $u$ into $u_1$. 
 If we have also $b \in \Gamb$ such that $u \equiv b u_2 \ov b$ and 
 $a\neq b$ then we have $(\alp, \bet) \in I$ and $u\equiv ab u_3\ov b\,  \ov a$.
 Thus the rewriting procedure is strongly confluent and therefore $p$ exists and the reduced dependence graph 
 of $\wt u \in \Gam^*$ is uniquely defined by $u$. 
  \end{proof}

We say that a reduced word $u$ is \emph{cyclically reduced} if 
the dependence graph of $u$ does not contain any minimal vertex $i$ and any 
maximal vertex $j$ such that $i \neq j$ but $\lam(i), \lam (j) \in \Gama$
for some $\alp \in L$. Thus, a reduced word $u$ is not cyclically reduced \IFF 
$u\equiv a u'a'$ for some $a,a' \in \Gama$. If $u$ is not cyclically reduced then $u\sim u'[a'a]$ and the length of  $u'[a'a]$ is shorter than $u$. 
Therefore it is enough to solve the \CP for cyclically reduced words, but the initial problem  is to compute them in \logs. The key observation 
to overcome this difficulty is the following lemma. 

\begin{lemma}\label{lem:fred}
Let $u\in \Gam^*$ be reduced. Then there are reduced words $p,r,m,s$ 
such that 
\begin{itemize}
\item $u \equiv prms\ov p$
\item $\abs{[sr]}_\alp =\abs{r}_\alp = \abs{s}_\alp \leq 1$ for all $\alp \in L$
\item $m[sr]$ is cyclically reduced and $u \sim m[sr]$. 
\end{itemize}
\end{lemma}

\begin{proof}
 Choose $p$ and $\wt u$ according to \prref{lem:cyc}. Next, there is a unique maximal $r$ such that 
 $\wt u\equiv rms$ with $\abs{[sr]}_\alp = \abs{r}_\alp \leq 1$ for all $\alp \in L$. This follows because $p$ has maximal length. Actually 
 for some subset $M \sse L$ of pairwise independent nodes 
 we have $r \equiv  \prod_{\alp\in M}a_\alp$ and $s \equiv  \prod_{\alp\in M}b_\alp$ such that $a_\alp$, $b_\alp \in \Gama$. Moreover 
 $[b_\alp a_\alp] \in \Gama$.  
 Thus, $m[sr] \equiv  m \prod_{\alp\in M}[b_\alp a_\alp] $ is reduced. 
 It is cyclically reduced because $r$ has been chosen to be maximal and 
 $u$ is reduced. 
 The assertion $u \sim m[sr]$ is trivial. 
 \end{proof}

\begin{lemma}\label{lem:nalja}
There is a \logs computation which on input $u \in \Gam^*$ outputs 
a cyclically reduced word $u'$ such that $u \sim u'$. 
\end{lemma}

\begin{proof}
 The idea is to compute a cyclically reduced word $ m[sr]$ with 
$u \sim m[sr]$  by reducing $w = uu$. 
Let us see what happens if we start the reduction process 
on $w= uu$ where, according to  \prref{lem:fred}, $u\equiv prms\ov p$ is reduced. We can write 
$w \equiv  prms\ov p  prms\ov p$. The word  $prm[sr]ms\ov p$ is a reduced word because $m[sr]$ and $[sr]m$ are cyclically reduced. Hence
$[w] \equiv  prm[sr]ms\ov p$. In order to determine the 
factor $m[sr]$ inside $[w]$ we compute for each 
$\alp \in L$ the $\alp$-lengths $\abs{w}_\alp$, $\abs{[w]}_\alp$,
$\abs{p}_\alp$, and $\abs{r}_\alp$. In a first \logs computation we 
determine the $\alp$-factorization of $u$. This gives us $n\in \N$ with $n= \abs{u}_\alp$
and therefore $2n =\abs{w}_\alp$. A second \logs computation using 
\prref{thm:compred} yields $k\in \N$ with $k= \abs{[w]}_\alp$.
Let $\eps =\abs{r}_\alp$. We know  
$\eps =\abs{[sr]}_\alp =\abs{r}_\alp = \abs{s}_\alp \leq 1$. We obtain 
\begin{align*}
2n &= 4\abs{p}_\alp + 2 \abs{m}_\alp + 4 \eps\\
k &= 2\abs{p}_\alp + 2 \abs{m}_\alp + 3 \eps
\end{align*}
 Thus, $2n -k = 2\abs{p}_\alp + \eps$. If $k$ is even then $\eps = 0$ 
 otherwise $\eps = 1$. Knowing $\eps= \abs{r}_\alp$ we know $\abs{p}_\alp$ and $\abs{m}_\alp$, too. We conclude that
 the $i$-th vertex of $D([w])$ which has a label in $\Gama$ belongs to the factor $m[sr]$  \IFF $$\abs{p}_\alp + \eps <i < k - \abs{p}_\alp- \eps - \abs{m}_\alp.$$ 
\end{proof}

\begin{lemma}\label{lem:hugo} Let $x,y\in \Gam^*$ be cyclically reduced such that 
$x \sim y$. Then $\abs{x}_\alp = \abs{y}_\alp$ for all $\alp \in L$. In particular, 
$\al(x) = \al(y)$.  
\end{lemma}
\begin{proof}
Let  $z\in \Gamma^*$ be reduced of minimal length such that $xz=zy \in G$. Assume by contradiction that there exists some $\alp\in \al(x) \sm \al(y)$. On the right hand side 
no reduction can involve letters from $\Gama$, hence $\abs{zy}_\alp = \abs{[zy]}_\alp = \abs{[z]}_\alp$, but $\abs{xz}_\alp \geq 1+ \abs{z}_\alp$.  Hence a reduction between $x$ and $z$ must occur. Hence there exist $a,a' \in \Gama$ such that we can write
 $x \equiv x'a$, $z \equiv a'z'$. If $aa' = 1 \in G$ then 
$ax'z'  = axz = azy = z'y$ and hence $ax'z' = z'y$. Since $x$ is cyclically reduced, the word $ax'$ is reduced, too.  
By induction on the length of $z$ we obtain 
$\abs{ax'}_\alp = \abs{y}_\alp$ and $\abs{y}_\alp \geq 1$ which is a contradiction. 
Thus, $aa' \neq  1 \in G$  and $[aa'] \in \Gama$.
Therefore, the $\alp$-length of $x'[aa']z'$ is equal to the $\alp$-length of $z = a' z'$. 
As $a'$ is minimal in $z$ we conclude $[aa'] = a'$, hence $a = 1 \in G$, which is again a contradiction. 
\end{proof}

For a subset $C \sse L$ let $G_C$ be the graph product which is defined with respect to the independence relation  $(C,I_C)$ where $I_C= I \cap C \times C$. 
Recall that $G_C$ is a retract of $G$ with respect to the canonical projection 
$\pi_C: G \to G_C$ since $\pi_C(g) = g$ for all $g \in G_C$. 
The following lemma shows that it is enough to decide conjugacy on connected words. 

\begin{lemma}\label{lem:gunnar}
Let $x,y\in \Gam^*$ be reduced such that  $\al(x) = \al(y)= A\cup B$ with $A\times B \subseteq I$.
Write $x = x_Ax_B, y=y_Ay_B$ with $\al(x_C)=\al(y_C) = C$ for $C\in \{A,B\}$. Then 
we have $ x\sim y $ in $G$ \IFF $x_C\sim y_C$ in $G$ for $C\in \{A,B\}$.
Moreover, $x_C\sim y_C$ in $G$ \IFF $x_C\sim y_C$ in $G_C$.
\end{lemma}
\begin{proof} 
Consider the canonical projection $\pi_C : G \to G_C$ and let $z_C= \pi_C(z)$.
If $xz = zy$, then $x_Cz_C = \pi_C(xz) = \pi_C(zy) = z_C y_C$. Hence $x_C\sim y_C$ in $G_C$. This implies  $x_C\sim y_C$ in $G$. Now, let $x_C\sim y_C$ in $G$ for $C\in \{A,B\}$. Choose $z'$ and $z''$ such that $x_Az' = z'y_A$ and $x_Bz'' = z''y_B$. 
We obtain $x_A\pi_A(z') = \pi_A(z')y_A$ and $x_B\pi_B(z'') = \pi_B(z'')y_B$.
It follows $xz = zy \in G$ for $z = \pi_A(z')\pi_B(z'')$. 
\end{proof}

We are now ready to prove the remaining implication of \prref{thm:conjug}. 
For this we may assume that the \CP in all $G_\alp$ is solvable in \logs. 
In order  to solve conjugacy in \logs for $G$,  it is enough 
to consider  cyclically reduced and connected input words $x$ and $y$ such that $\al(x) = \al(y)$. Let us consider the special case where $\abs{\al(x)} \leq  1$ first. Then we have $x$, $y \in \Ga$ for some $\alp \in L$. Another consequence of 
\prref{lem:gunnar} is that now $x \sim y$ in $G$ \IFF $x \sim y$ in $\Ga$. This is the only place where we use that the conjugacy problem is solvable in \logs for all 
node groups. Thus, we may assume that  $\abs{\al(x)} \geq 2$. 
This leads us to combinatorics on dependence graphs in the spirit of 
\cite{dr95}. 

Let us define the notion of \emph{transposition}. We say that 
words $u,v  \in \Gam^*$ are \emph{transposed} if there are 
$r,s \in \Gam^*$ such that $u\equiv rs$ and $v \equiv sr$.
Thus, the definition is based on dependence graphs. Transposition is a reflexive and symmetric 
relation. But unlike the usual definition for words it is not transitive, in general. (The usual definition is the special case where $I$ is empty.) By $u \approx v$ we denote the transitive closure 
of transposition. We can view $\approx$ as an equivalence relation on dependence graphs. A crucial observation is that if $u \approx v$ 
and $u$ is cyclically reduced then $v$ is cyclically reduced, too. 
In the following, we consider $\approx$ only for cyclically reduced words. 
Clearly, if $u\approx v$, then $u\sim v$ in $G$, but the converse does not hold in general. To see this let $a \sim a'$ in some $\Ga$ with 
$a \neq a'$. Then $a,a' \in \Gama$ are cyclically reduced, but $a \not\approx a'$. Using transpositions on cyclically reduced words we never can multiply letters together which are neighbors in the 
Hasse diagram and we obtain an analogue to  Duboc's classical result  which characterizes $\approx$ for partially commutative monoids \cite{dub86tcs1}. 

\begin{proposition}\label{prop:duboc}
Let  $u,v  \in \Gam^*$ be cyclically reduced words. Then we have 
$u \approx v$ \IFF the following two conditions hold. First, 
we have ${\abs u}_\alp =  {\abs v}_\alp$ for all $\alp \in L$ and second, there are reduced words $p,q$ such that
$puq \equiv v^{\abs L}$. 
\end{proposition}

\begin{proof}
Duboc's result \cite{dub86tcs1} (see also  \cite[Thm.~3.3.3]{dr95})  is stated for  Mazurkiewicz traces. It can be applied because $u,v  \in \Gam^*$ are cyclically reduced. Actually 
her proof can be applied verbatim in our setting. 

First, let $u \approx v$. 
It follows ${\abs u}_\alp =  {\abs v}_\alp$ for all $\alp \in L$. We may assume $u \not\equiv v$ and we use induction on the number of transpositions to transform $v$ into  $u$. Since $u \not\equiv v$ there are $r$, $s$ such that 
$v \equiv rs$ and  such that the number of transpositions to transform $sr$ into  $u$ has decreased. By induction, there are reduced words $p',q'$ such that
$p'uq' \equiv (sr)^k$ for some $k \in \N$. Let $p =rp'$ and $q= q's$ then 
we see $puq \equiv (rs)^{k+1}$.  It remains to show that 
we can bound the exponent $k$ by $\abs L$. To see this let 
$puq \equiv v_1 \cdots v_k$ for some $k$ such that each $v_\ell \equiv v$.
Without restriction, $u$ is connected.
A minimal vertex $i_0$ of $u$ must be located in  $v_1$. Now, for a vertex $j$ 
of $u$ we let $d(i_0,j)$ be the length of a shortest path from $i_0$ to $j$ 
in the dependence graph $D(u)$. We claim, that if  $d = d(i_0,j)$, then 
$j$ appears as a vertex in the prefix $v_1 \cdots v_{d+1}$. 
The claim follows by induction on $d$. Let $i$ be a vertex of $u$ which appears 
in $v_1 \cdots v_{d}$ and $\lam(i)\in \Gama$, $\lam(j)\in \Gamb$, with $(\alp, \bet) \in D$. On a path from  $i$ to $j$ in $v_1 \cdots v_k$ there are at most 
${\abs u}_\bet$ vertices with a label in $\Gamb$. Since ${\abs u}_\bet
= {\abs v}_\bet$, we conclude the claim. Since always $d \leq \abs L -1$, 
we obtain that there are reduced words $p,q$ such that
$puq \equiv v^{\abs L}$.

For the other direction let ${\abs u}_\alp =  {\abs v}_\alp$ for all $\alp \in L$ and   $p,q$ be reduced words such that
$puq \equiv v^k$ for some $k \in \N$. If we have $\abs p = 0$ then 
$u\equiv v$ since ${\abs u}_\alp =  {\abs v}_\alp$ for all $\alp \in L$.
Thus we have $p \equiv a p'$ for some $a \in \Gama$ and $p' \in \Gam^*$. We conclude $v \equiv a v'$ for some $v' \in \Gam^*$. This leads to 
$p'uqa \equiv (v'a)^k$. By induction on the length of $p$ we obtain 
$u \approx (v'a) \approx (av') \equiv v$. Hence the result.  
  \end{proof}

\begin{corollary}\label{cor:lisa}
Let $G= G(L,I; (\Ga)_{\alp\in L}))$ 
be a graph product and $\GG$ as above
be such that  the \WP of all node groups $\Ga$ 
is in 
\logs. 
Then the following problem can be solved in \logs. 
Input: Cyclically reduced  words $u$, $v \in \Gam^*$. 
Problem: 
Do we have  $u \approx v$?
\end{corollary}

\begin{proof}
This is a direct consequence of \prref{prop:patmat}
and \prref{prop:duboc}.
\end{proof}

Using \prref{cor:lisa}  the proof of  
\prref{thm:conjug} is reduced to showing the following 
combinatorial proposition.

\begin{proposition}\label{prop:cot}
Let $G= G(L,I; (\Ga)_{\alp\in V}))$ be a graph product, $\GG$ as above, and let $x,y\in \GG^*$ be  cyclically reduced and connected words such that $\al(x) = \al(y)$ with $\abs{\al(x)} \geq 2$.  Then we have 
$x \sim y $ in the group $G$ \IFF $x \approx y $.
\end{proposition}

\begin{proof}
Let $x \sim y $. We have to show $x \approx y $.
Choose some reduced word $z \in \Gam^*$ of minimal length such that 
$xz = zy$ in $G$. By \prref{lem:gunnar} we have $\al z \sse \al x$. 
Assume that $xz$ was not reduced. Then we have $x\equiv x'a$ and $z \equiv a'z'$
such that $[aa'] \in \Ga$ for some $\alp \in L$. Since $a'$ is minimal in $z$ and $\al z \sse \al x$ we conclude that $a'$ is also minimal in $x$.
Actually there is a minimal vertex in $x'$ with label $a'$, because $x$ is connected and $\abs{\al(x)} \geq 2$. This implies $x= a'x''a$ which  
is a contradiction since $x$ is cyclically reduced. 
 
Thus, $xz$ is reduced and therefore $zy$ is reduced, too. This implies
$xz \equiv zy$ because $xz = zy$ in $G$. We can apply the Levi Lemma for traces \cite[Thm. 3.2.2]{dr95}. 
It yields the existence of $p,r,s,q \in \Gam^*$ such that 
$x \equiv pr, z\equiv sq, z\equiv ps, y \equiv rq$ with 	$r$ and $s$ independent.
If $\abs p = 0$ or $\abs q = 0$ then $x=y$, hence we may assume 
$\abs s < \abs z$. Moreover, $rps \equiv rsq \equiv s rq$ because $rs \equiv sr$. Thus, by induction we obtain $rp \approx rq \equiv y$. Now, $rp $ and $x \equiv pr$ are transposed, hence  $x \approx y $.
\end{proof}

We now have all the ingredients to describe the algorithm which proves 
\prref{thm:conjug} (and \prref{cor:conjug}). \\

{\noindent \bf The Algorithm for solving conjugacy in a graph product.} \\
{\bf Input:}  $u$, $v \in \GG^*$. {\bf Question:} $u \sim v $ in $G$? 
\begin{enumerate}
\item Compute $u$, $v$ in reduced form using \prref{thm:compred}.
\item Compute $u$, $v$ in cyclically reduced form using \prref{lem:fred}.
\item Reduce to the case that $u$, $v$ are cyclically reduced and connected using \prref{lem:gunnar}.
\item Compute $\al(u)$, $\al(v)$. If $\al(u)\neq \al(v)$ then 
$u \not\sim v $ in $G$. Hence without restriction, $\al(u) = \al(v)$.
\item If $\abs{\al(u)} \leq 1$ then $u \sim v $ in $G$ \IFF $u$ and 
$v$ are conjugated in the corresponding node group. Hence without restriction, $\abs{\al(u)} \geq 2$.
\item Since $\abs{\al(u)} \geq 2$, we have now by \prref{prop:cot} that $u \sim v $ in $G$ \IFF $u \approx v$. Decide  $u \approx v$ using \prref{cor:lisa}.
\end{enumerate} 

\section{Conclusion}\label{sec:concl}
The paper shows transfer results for the \logs complexity of important
group-theoretical decision problems from node groups to graph products. 
This concerns the word problem, computing geodesics,  and the \CP. 
The first two results were  known for RAAGs (graph groups) before, but
not for graph products in general. The earlier proof for RAAGs relied on  geometry and linear representations and these methods are not available for 
graph products, in general. The present proof is purely combinatorial. 
Our results  concerning the \CP are new even for RAAGs, and they  go clearly far beyond that. Our results also support a conjecture that 
a graph product of linear groups is again linear. 
A proof of this  conjecture might proceed using a similar induction scheme as used here, but this is highly speculative and not in the scope of 
purely combinatorial methods. 
 
An  interesting question is whether analogous transfer results hold 
in complexity classes below \logs. 
The main obstacle is the \WP for free groups in two generators.
The precise complexity of this  \WP is  a long standing open question in algorithmic 
group theory. 

A promising line of future research is  to  extend the results beyond graph products. 
For example, the results in \prref{sec:amal} can easily be extended from direct products to semi-direct products. But this is only the first step. 


%
%
%
%
%
%
%
%
%
%
%


\begin{thebibliography}{10}

\bibitem{BesBr97}
M.~Bestvina and N.~Brady.
\newblock Morse theory and finiteness properties of groups.
\newblock {\em Inventiones Mathematicae}, 129(3):445--470, 1997.


\bibitem{cai92stoc}
J.-Y. Cai.
\newblock Parallel computation over hyperbolic groups.
\newblock In {\em Proc. 24th ACM Symp. on Theory of Computing, STOC 92}, pages
  106--115. ACM-press, 1992.

\bibitem{CrGoWi09}
J.~Crisp, E.~Godelle, and B.~Wiest.
\newblock The conjugacy problem in right-angled {{A}rtin} groups and their
  subgroups.
\newblock {\em Journal of Topology}, 2(3):442--460, 2009.

\bibitem{CrWi04}
J.~Crisp and B.~Wiest.
\newblock Embeddings of graph braid and surface groups in right-angled {A}rtin
  groups and braid groups.
\newblock {\em Algebraic \& Geometric Topology}, 4:439--472, 2004.

\bibitem{dkl12conm}
V.~Diekert, J.~Kausch, and M.~Lohrey.
\newblock Logspace computations in {C}oxeter groups and graph groups.
\newblock In {\em Computational and Combinatorial Group Theory and
  Cryptography}, volume 582 of {\em Contemporary Mathematics}, pages 77--94.
  Amer. Math. Soc., 2012.
\newblock Journal version of LATIN 2012, 243--254, LNCS 7256, 2012.

\bibitem{dl08ijac}
V.~Diekert and M.~Lohrey.
\newblock Word equations over graph products.
\newblock {\em International Journal of Algebra and Computation}, 18:493--533,
  2008.
\newblock Conference version in FSTTCS 2003.

\bibitem{dm06}
V.~Diekert and A.~Muscholl.
\newblock Solvability of equations in free partially commutative groups is
  decidable.
\newblock {\em International Journal of Algebra and Computation},
  16:1047--1070, 2006.
\newblock Journal version of ICALP 2001, 543--554, LNCS 2076.

\bibitem{dr95}
V.~Diekert and G.~Rozenberg, editors.
\newblock {\em The Book of Traces}.
\newblock World Scientific, Singapore, 1995.

\bibitem{dro85}
C.~Droms.
\newblock Graph groups, coherence and three-manifolds.
\newblock {\em Journal of Algebra}, 106(2):484--489, 1985.

\bibitem{DLS}
C.~{Droms}, J.~{Lewin} and H.~{Servatius}.
\newblock {The length of elements in free solvable groups}.
\newblock {\em Proc. Amer. Math. Soc.}, 119:27--33, 1993.


\bibitem{dub86tcs1}
{\Ch}.~Duboc.
\newblock On some equations in free partially commutative monoids.
\newblock {\em Theoretical Computer Science}, 46:159--174, 1986.

\bibitem{elder2010}
M.~Elder.
\newblock A linear-time algorithm to compute geodesics in solvable
  {B}aumslag-{S}olitar groups.
\newblock {\em Illinois Journal of Mathematics}, 54(1):109--128, 2010.

\bibitem{ElderEO13}
M.~Elder, G.~Elston, and G.~Ostheimer.
\newblock On groups that have normal forms computable in logspace.
\newblock {\em J. Algebra}, 381:260--281, 2013.

\bibitem{ElRe2010}
M.~Elder and A.~Rechnitzer.
\newblock Some geodesic problems in groups.
\newblock {\em Groups. Complexity. Cryptology}, 2(2):223--229, 2010.


\bibitem{HauboldLM12}
N. Haubold, M. Lohrey and C. Mathissen.
\newblock {Compressed Decision Problems for Graph Products and Applications to (outer) Automorphism Groups}.
\newblock {\em IJAC}, 22, 2012.

\bibitem{hw99}
T.~Hsu and D.~T. Wise.
\newblock On linear and residual properties of graph products.
\newblock {\em Michigan Mathematical Journal}, 46(2):251--259, 1999.

\bibitem{kel73}
R.~M. Keller.
\newblock Parallel program schemata and maximal parallelism~{I}. {F}undamental
  results.
\newblock {\em Journal of the Association for Computing Machinery},
  20(3):514--537, 1973.

\bibitem{lz77}
R.~J. Lipton and Y.~Zalcstein.
\newblock Word problems solvable in logspace.
\newblock {\em Journal of the Association for Computing Machinery},
  24(3):522--526, 1977.

\bibitem{Lo05ijfcs}
M.~Lohrey.
\newblock Decidability and complexity in automatic monoids.
\newblock {\em International Journal of Foundations of Computer Science},
  16(4):707--722, 2005.

\bibitem{Lohrey2012survey}
M.~Lohrey.
\newblock Algorithmics on {SLP}-compressed strings: {A} survey.
\newblock {\em Groups Complexity Cryptology}, 4:241--299, 2012.

\bibitem{maz77}
A.~Mazurkiewicz.
\newblock Concurrent program schemes and their interpretations.
\newblock {DAIMI Rep. PB}~78, Aarhus University, Aarhus, 1977.

\bibitem{Miller1}
C.~F. {Miller III}.
\newblock {\em {On group-theoretic decision problems and their
  classification}}, volume~68 of {\em Annals of Mathematics Studies}.
\newblock Princeton University Press, 1971.

\bibitem{Miller}
C.~F. {Miller III}.
\newblock {Decision problems for groups -- survey and reflections}.
\newblock {\em Springer}, 1992.

\bibitem{MyRoUsVe08}
A.~Myasnikov, V.~Roman'kov, A.~Ushakov, and A.~Vershik.
\newblock The word and geodesic problems in free solvable groups.
\newblock {\em Transactions of the American Mathematical Society},
  362:4655--4682, 2010.

\bibitem{papa}
{\Ch}.~{Papadimitriou}.
\newblock {\em Computation Complexity}.
\newblock Addison-Wesley, 1994.

\bibitem{PRaz}
M.~{Paterson} and A.~{Razborov}.
\newblock {The set of minimal braids is co-NP-complete}.
\newblock {\em J. Algorithms}, 12:393--408, 1991.

\bibitem{Robinson93phd}
D.~Robinson.
\newblock {\em Parallel Algorithms for Group Word Problems}.
\newblock PhD thesis, University of California, San Diego, 1984.

\bibitem{Sapir2011BullMath}
M.~Sapir.
\newblock Asymptotic invariants, complexity of groups and related problems.
\newblock {\em Bulletin of Mathematical Sciences}, 1(2):277--364, 2011.

\bibitem{serre80}
J.-P. Serre.
\newblock {\em Trees}.
\newblock Springer, 1980.
\newblock French original 1977.

\bibitem{Sim79}
H.-U. Simon.
\newblock Word problems for groups and contextfree recognition.
\newblock In {\em Proceedings of Fundamentals of Computation Theory (FCT'79),
  Berlin/Wendisch-Rietz (GDR)}, pages 417--422. Akademie-Verlag, 1979.

\bibitem{Vollmer99}
H.~Vollmer.
\newblock {\em Introduction to Circuit Complexity}.
\newblock Springer, Berlin, 1999.

\bibitem{Waack81}
S.~Waack.
\newblock Tape complexity of word problems.
\newblock In F.~G{\'e}cseg, editor, {\em Proceedings of Fundamentals of
  Computation Theory (FCT'81)}, volume 117 of {\em Lecture Notes in Computer
  Science}, pages 467--471. Springer, 1981.

\bibitem{wehrfritz73lms}
B.~A.~F. {Wehrfritz}.
\newblock {Generalized free products of linear groups}.
\newblock {\em Proc. LMS}, 27:402--424, 1973.

\end{thebibliography}
\newcommand{\Ju}{Ju}\newcommand{\Ph}{Ph}\newcommand{\Th}{Th}\newcommand{\Ch}{Ch}\newcommand{\Yu}{Yu}\newcommand{\Zh}{Zh}\newcommand{\St}{St}

\end{document}